\documentclass[11pt]{amsart}
\usepackage{epsfig}
\usepackage{amsbsy,latexsym}
\usepackage{amsmath}
\usepackage{amssymb, mathrsfs}
\usepackage[mathscr]{eucal}
\usepackage{hyperref}
\usepackage{amsfonts}
\usepackage{amsthm}
\usepackage{amsmath}
\usepackage{amsfonts}
\usepackage{latexsym}
\usepackage{amssymb}
\usepackage{amscd}
\usepackage[latin1]{inputenc}
\usepackage{verbatim}

\usepackage[english]{babel}

\newtheorem{definition}{Definition}[section]
\newtheorem{theorem}[definition]{Theorem}
\newtheorem{lemma}[definition]{Lemma}
\newtheorem{corollary}[definition]{Corollary}
\newtheorem{proposition}[definition]{Proposition}
\theoremstyle{definition}
\newtheorem{remark}[definition]{Remark}

\newtheorem{defn}[definition]{Definition}
\newtheorem{example}[definition]{Example}

\newcommand{\C}{\mathbb{C}}

\newcommand\tr{ \operatorname{Tr} }

\def\ket#1{| #1 \rangle}
\def\bra#1{\langle #1 |}
\def\kb#1#2{|#1\rangle\!\langle #2 |}

\title[Approximate Quasiorthogonality and Relative Quantum Privacy]{Approximate Quasiorthogonality of Operator Algebras and Relative Quantum Privacy}

\begin{document}

\author[D.~W. Kribs, J.~Levick, M.~Nelson, R. Pereira, M.~Rahaman]{David~W.~Kribs$^{1,2}$, Jeremy Levick$^{1,2}$, Mike Nelson$^{1}$, Rajesh Pereira$^{1}$, Mizanur Rahaman$^{2,3}$}

\address{$^1$Department of Mathematics \& Statistics, University of Guelph, Guelph, ON, Canada N1G 2W1}
\address{$^2$Institute for Quantum Computing, University of Waterloo, Waterloo, ON, Canada N2L 3G1}
\address{$^3$Department of Pure Mathematics, University of Waterloo, Waterloo, ON, Canada N2L 3G1}

\subjclass[2010]{47L90, 81P45, 81P94, 94A40}

\keywords{completely positive map, private quantum channel, conditional expectation, operator algebra, private quantum code, quasiorthogonal algebra.}


\begin{abstract}
We show that the approximate quasiorthogonality of two operator algebras is equivalent to the algebras being approximately private relative to their conditional expectation quantum channels. Our analysis is based on a characterization of the measure of orthogonality in terms of Choi matrices and Kraus operators for completely positive maps. We present examples drawn from different areas of quantum information.
\end{abstract}

\maketitle

\section{Introduction}

The notion of quasiorthogonality for operator algebras arose from the study of modified forms of orthogonality for algebras and their relative behaviours in a variety of settings in finite-dimensional quantum information. Primarily motivated by mutually unbiased bases (MUB) constructions  \cite{bandyopadhyay2002new,kr04,kimura2006,durt2010mutually,spengler2012entanglement} and their associated commutative algebras initially, over the past decade the work expanded to the non-commutative setting \cite{petz2007complementarity,ohno2007quasi,petz2007complementary,ohno2008quasi,petz2009complementarity,petz2010algebraic}. Built on these efforts, Weiner \cite{Wein} opened up the study of approximate quasiorthogonality by introducing a measure of orthogonality between two algebras based on joint properties of their conditional expectation channels, and establishing results on the approximate version for some important special cases.

From a different direction, still with quantum information motivation, it has recently been recognized that there are connections between the study of quasiorthogonal operator algebras and work in quantum privacy; specifically, on the topic of what are variously known as private quantum channels or codes, decoherence-full or private subspaces and subsystems, and private algebras \cite{ambainis,boykin,bartlett1,bartlett2,kks,church,jochym,jochym1,cklt}. In particular, for a number of special cases of channels or algebras, quasiorthogonality has been linked with certain quantum privacy properties in those cases \cite{ljklp,klp17,kribs2018quantum}, all suggesting a deeper more general link between the topics.

In this paper, we establish the first general result that ties together approximate quasiorthogonality of operator algebras with approximate privacy for quantum codes, represented as algebras. This involves identification of an appropriate notion of relative quantum privacy, with natural assumptions on the algebras and private quantum codes considered, and in doing so, we also derive a new approach for computing Weiner's measure of orthogonality in terms of Choi matrices and Kraus operators for the conditional expectation channels of the associated algebras. We present examples drawn from the framework for hybrid classical-quantum information theory, from studies of private quantum subsystems, and from work on approximate MUB constructions.

The paper is organized as follows. The next section includes preliminary material and the derivation of our approach to compute the orthogonality measure. In the third section we define relative approximate privacy of two algebras and present our main result and its proof. The fourth section contains examples and we conclude with a brief outlook discussion.

\section{Measure of Quasiorthogonality and Quantum Privacy}

Let $M_n(\C)$ be the set of $n\times n$ complex matrices. In all of what follows, $\mathcal{A},\mathcal{B}\subseteq M_n(\C)$ are unital $*$-algebras (or finite-dimensional C$^*$-algebras \cite{davidson}); that is, $I_n\in \mathcal{A}$, if $A \in \mathcal{A}$ so is $A^*$, and $\mathcal{A}$ is closed under linear combinations and matrix multiplication, and the same is true for $\mathcal{B}$.

Given such $\mathcal{A}$, $\mathcal{B}$ we denote by $\mathcal{E}_{\mathcal{A}}$, $\mathcal{E}_{\mathcal{B}}$ the (unique) trace-preserving, unital conditional expectations onto $\mathcal{A}$ and $\mathcal{B}$ respectively. That is, $\mathcal{E}_\mathcal{A}:M_n(\C)\rightarrow M_n(\C)$ (and similarly for $\mathcal E_{\mathcal B}$) is the linear map uniquely characterized by the following conditions:
\begin{enumerate}
\item $\mathcal{E}_{\mathcal{A}}(A) = A$ for all $A\in \mathcal{A}$
\item $\mathcal{E}_{\mathcal{A}}(X) \succeq 0$ whenever $X\succeq 0$
\item $\mathcal{E}_{\mathcal{A}}(A_1 X A_2) = A_1\mathcal{E}_{\mathcal{A}}(X)A_2$ for all $A_1,A_2\in \mathcal{A}$ and $X\in M_n(\C)$
\item $\mathrm{Tr}(\mathcal{E}_{\mathcal{A}}(X)) = \mathrm{Tr}(X)$ for all $X\in M_n(\C)$.
\end{enumerate}
We refer to $\mathcal E_{\mathcal A}$ as the {\it conditional expectation channel} for $\mathcal A$, reflecting standardized use of the term quantum channel to describe completely positive trace-preserving maps.

We now define the key notion of quasiorthogonal algebras, noting the early literature on the subject sometimes referred to the notion as `orthogonal' or `complementary'. (We stick to use of the `quasi' prefix as in \cite{Wein} as it avoids possible confusion with other quantum information notions that use these other terms.)

\begin{defn}\label{q.o} Two unital $*$-algebras $\mathcal{A}$, $\mathcal{B}$ are \emph{quasiorthogonal} if they satisfy any one of the following equivalent conditions:
\begin{enumerate}
\item $\mathrm{Tr}\bigl((A - \frac{\mathrm{Tr}(A)}{n}I_n)(B-\frac{\mathrm{Tr}(B)}{n}I_n)\bigr) = 0$ for all $A\in \mathcal{A}$, $B\in \mathcal{B}$
\item $\mathrm{Tr}(AB) = \frac{\mathrm{Tr}(A)\mathrm{Tr}(B)}{n}$ for all $A\in \mathcal{A}$, $B\in \mathcal{B}$
\item\label{3} $\mathcal{E}_{\mathcal{A}}(B) = \frac{\mathrm{Tr}(B)}{n}I_n$ for all $B\in \mathcal{B}$ and $\mathcal{E}_{\mathcal{B}}(A) = \frac{\mathrm{Tr}(A)}{n}I_n$ for all $A \in \mathcal{A}$
\end{enumerate}
\end{defn}

The ideal ($\epsilon = 0$ -- see next section) notion of quantum privacy we consider here is given as follows.

\begin{defn} Given a unital quantum channel $\Phi:M_n(\C)\rightarrow M_n(\C)$, we say a unital $*$-algebra $\mathcal{A}$ is {\it private for $\Phi$} whenever
$$\Phi(A) = \frac{\mathrm{Tr}(A)}{n}I_n, $$ for all $A\in \mathcal{A}$.
\end{defn}

\begin{remark}
Notice that Condition \ref{3} from Definition \ref{q.o} asserts that if $\mathcal{A}$ and $\mathcal{B}$ are quasiorthogonal then the conditional expectation onto the one algebra privatizes the other. The simplest example of this phenomena can be seen in the extreme case with $\mathcal A = M_n(\mathbb{C})$ and $\Phi = \mathcal D_n$ is the `complete depolarizing' channel, $\mathcal D_n(X) = n^{-1} \mathrm{Tr}(X) I_n$, which is the conditional expectation channel onto the scalar algebra $\mathcal B = \mathbb{C} I_n$. Observe also in this case that $\mathcal A$ and $\mathcal B$ are quasiorthogonal, which is a simple special case of our main result below.

More general notions of private algebras have been considered in the literature, often with different nomenclature as well, such as private quantum channels, decoherence-full or private subspaces and subsystems, and private algebras   \cite{ambainis,boykin,bartlett1,bartlett2,kks,church,jochym,jochym1,cklt,ljklp,klp17}. The distinguished special case we consider here captures many of the most naturally occurring examples from these settings, in addition to, as we shall see, allowing us to establish a tight connection with quasiorthogonality in the approximate case.
\end{remark}

In \cite{Wein}, Weiner introduced the following quantitative measure of orthogonality for algebras. We will focus on this notion for the rest of the section.

\begin{defn} For $\mathcal{A},\mathcal{B}\subseteq M_n(\C)$ unital $*$-algebras, the \emph{measure of orthogonality} between them is given by
\begin{equation}\label{mq1}\mathcal{Q}(\mathcal{A},\mathcal{B}):=\mathrm{Tr}(T_{\mathcal{A}}T_{\mathcal{B}})\end{equation} where $T_{\mathcal{A}}$ is the (any) matrix representation of $\mathcal{E}_{\mathcal{A}}$ acting on the vector space $M_n(\C)$.
\end{defn}

We shall make use of the explicit forms of our matrix representations so let us introduce notation $\{ \ket{i} : 1 \leq i \leq n \}$ for a fixed orthonormal (o.n.) basis for $\C^n$, and then $E_{ij} = \ket{i}\bra{j}$ for the corresponding set of matrix units of $M_n$, which themselves form an orthonormal basis in the trace inner product; $<A,B> = \tr (B^*A)$, $A,B\in M_n$. We will then work with the so-called `natural representation' $T_\Phi$ \cite{watrous} for a given linear map $\Phi : M_n(\C) \rightarrow M_n(\C)$ in the basis $\{ E_{ij} \otimes E_{kl}  \}$ this set defines for $M_{n^2}(\C) \cong M_n(\C) \otimes M_n(\C)$;
\[
T_{\Phi} = \sum_{i,j,k,l} < \Phi(E_{kl}) , E_{ij} > E_{ik} \otimes E_{jl} .
\]
On the other hand, we can consider the Choi matrix \cite{choimtx} for $\Phi$, given by
\[
C_{\Phi}:=\sum_{i,j=1}^n E_{ij}\otimes \Phi(E_{ij}).
\]
Using the expansion $\Phi(E_{ij}) = \sum_{k,l} < \Phi(E_{ij}) , E_{kl} > E_{kl}$, we see that $C_\Phi$ and $T_\Phi$ have the same matrix coefficients up to the (unitarily implemented) permutation that sends $E_{ij}\otimes E_{kl}\mapsto E_{ik}\otimes E_{jl}$.

Hence, if we denote by $C_{\mathcal{A}}$ the Choi matrix of $\mathcal{E}_{\mathcal{A}}$, and similarly for $C_{\mathcal{B}}$, then we have the following observation.

\begin{proposition}
Given $*$-algebras $\mathcal A$ and $\mathcal B$, we have
\begin{equation}\label{mq2}\mathcal{Q}(\mathcal{A},\mathcal{B}) = \mathrm{Tr}(C_{\mathcal{A}} C_{\mathcal{B}}).\end{equation}
\end{proposition}

We shall make use of the following internal description of the Choi matrix for a conditional expectation. Recall that if $A$ is a matrix, $\overline{A}$ is its complex conjugate matrix.

\begin{lemma} Let $\mathcal{A}\subseteq M_n(\C)$ be a $*$-algebra and let $\{A_i\}_{i=1}^{d_1}$ be any o.n. basis for $\mathcal{A}$. Then we have
\begin{equation}\label{choiform} C_{\mathcal{A}} = \sum_{i=1}^{d_1} \overline{A_i}\otimes A_i.\end{equation}
\end{lemma}

\begin{proof} We first extend $\{A_i\}_{i=1}^{d_1}$ to an o.n. basis on the full matrix space by appending the elements of $\{A^{\perp}_j\}_{j=1}^{n^2 - d_1}$.
We claim that \begin{equation}\label{basis} \sum_{i=1}^{d_1} \overline{A_i}\otimes A_i + \sum_{j=1}^{n^2 - d_1}\overline{A_j^{\perp}}\otimes A_j^{\perp} = \sum_{i,j=1}^n E_{ij}\otimes E_{ij}.\end{equation}
To see why this is so, observe that for all $X \in M_n(\C)$,
\[
(\tr \otimes \mathrm{id})\biggl((X^T\otimes I_n)(\sum_{i=1}^{d_1}\overline{A_i}\otimes A_i + \sum_{j=1}^{n^2 -d_1} \overline{A_j^{\perp}}\otimes A_j^{\perp})\biggr)
\]
\begin{align*}& = \sum_{i=1}^{d_1}\mathrm{Tr}(X^T\overline{A_i})A_i + \sum_{j=1}^{n^2-d_1}\mathrm{Tr}(X^T \overline{A_j^{\perp}})A_j^{\perp}\\
& = \sum_{i=1}^{d_1} \mathrm{Tr}(XA_i^*)A_i + \sum_{j=1}^{n^2-d_1} \mathrm{Tr}(XA_j^{\perp *})A_j^{\perp}\\
& = \sum_{i=1}^{d_1} <X , A_i > A_i + \sum_{j=1}^{n^2-d_1} <X , A_j^{\perp *} > A_j^{\perp}\\
& = X,
\end{align*}
since $\{A_i\}\cup\{A_j^{\perp}\}$ forms an o.n. basis for $M_n(\C)$.
If we denote $B = \sum_{i=1}^{d_1} \overline{A_i}\otimes A_i + \sum_{j=1}^{n^2 - d_1}\overline{A_j^{\perp}}\otimes A_j^{\perp}$, then we observe that the property encoded above is $(\tr \otimes \mathrm{id})\bigl((X^T\otimes I_n)B \bigr) = X$, which uniquely characterizes $B$ as the Choi matrix of the identity map \cite{watrous}, and so $B = \sum_{i,j=1}^n E_{ij}\otimes \mathrm{id}(E_{ij}) = \sum_{i,j=1}^n E_{ij}\otimes E_{ij}$, as claimed.

Now, we recall that $C_{\mathcal{A}}:= (\mathrm{id}\otimes \mathcal{E}_{\mathcal{A}})\bigl(\sum_{i,j=1}^n E_{ij}\otimes E_{ij}\bigr)$ and so by Eq.~(\ref{basis}) we can equivalently say that
$$C_{\mathcal{A}} = (\mathrm{id}\otimes \mathcal{E}_{\mathcal{A}})\bigl( B \bigr) =  \sum_{i=1}^{d_1} \overline{A_i} \otimes \mathcal{E}_{\mathcal{A}}(A_i) + \sum_{j=1}^{n^2-d_1}\overline{A_j^{\perp}}\otimes \mathcal{E}_{\mathcal{A}}(A_j^{\perp}).$$
Since $\mathcal{E}_{\mathcal{A}}(A) = A$ for all $A\in \mathcal{A}$, and $\mathcal{E}_{\mathcal{A}}(A_j^{\perp}) = 0$ for each $A_j^{\perp}$, this simplifies to
$$C_{\mathcal{A}} = \sum_{i=1}^{d_1} \overline{A_i}\otimes A_i,$$
completing the proof.
\end{proof}

This leads to the following 2-norm type characterization of $Q(\mathcal A, \mathcal B)$.

\begin{corollary}\label{Qcompute}
Let $\mathcal{A}$, $\mathcal{B}\subseteq M_n(\C)$ be unital $*$-algebras. Then the measure of orthogonality may equivalently be expressed as:
\begin{equation}\label{mq3} \mathcal{Q}(\mathcal{A},\mathcal{B})=\sum_{i,j=1}^{d_1,d_2}|\mathrm{Tr}(A_iB_j)|^2, \end{equation} for any o.n. bases $\{A_i\}_{i=1}^{d_1}$, $\{B_j\}_{j=1}^{d_2}$ for $\mathcal{A}$ and $\mathcal{B}$ respectively.
\end{corollary}

\begin{proof}
We apply the formula of Eq.~(\ref{mq2}) to the expression for $C_{\mathcal{A}}$ obtained above to obtain:
\begin{align*}
\mathcal{Q}(\mathcal{A},\mathcal{B}) &= \mathrm{Tr}(C_{\mathcal{A}}C_{\mathcal{B}}) \\
& = \sum_{i,j=1}^{d_1,d_2}\mathrm{Tr}\bigl((\overline{A_i}\otimes A_i)(\overline{B_j}\otimes B_j)\bigr)\\
&= \sum_{i,j=1}^{d_1,d_2} |\mathrm{Tr}(A_iB_j)|^2.
\end{align*}
\end{proof}

We have stated that $\mathcal{Q}(\mathcal{A},\mathcal{B})$ is a measure of orthogonality. This is most explicitly seen through the following elementary observation of Weiner \cite{Wein}. We give an alternate simple proof based on the descriptions derived here.

\begin{proposition} Let $\mathcal{A}$, $\mathcal{B}\subseteq M_n(\C)$ be unital $*$-algebras. Then $\mathcal{A}$ and $\mathcal{B}$ are quasiorthogonal to one another if and only if $\mathcal{Q}(\mathcal{A},\mathcal{B}) = 1$.
\end{proposition}

\begin{proof}
Let $\{A_i\}_{i=1}^{d_1}$, $\{B_j\}_{j=1}^{d_2}$ be o.n. bases for $\mathcal{A}$, $\mathcal{B}$ respectively, and suppose $A_1 = B_1 = \frac{1}{\sqrt{n}}I_n$ (which is possible since both algebras are unital), so that $\mathrm{Tr}(A_i) = \mathrm{Tr}(B_j) = 0$ for all $1<i\leq d_1$, $1< j \leq d_2$.

Then if $\mathcal A$ and $\mathcal B$ are quasiorthogonal, and using the fact that $A_i, B_j$ are traceless for $i,j\neq 1$, we have
\begin{align*}\mathcal{Q}(\mathcal{A},\mathcal{B}) &= \sum_{i,j=1}^{d_1,d_2} |\mathrm{Tr}(A_iB_j)|^2   \\
& = \sum_{i,j=1}^{d_1,d_2} \frac{1}{n^2}|\mathrm{Tr}(A_i)|^2|\mathrm{Tr}(B_j)|^2   \\
& = \frac{1}{n^2}\bigl|\frac{\mathrm{Tr}(I_n)}{\sqrt{n}}\bigr|^2\bigl|\frac{\mathrm{Tr}(I_n)}{\sqrt{n}}\bigr|^2  \\
& = 1. \\
\end{align*}

Conversely, keeping our o.n. bases for $\mathcal{A}$ and $\mathcal{B}$, suppose that
$$\mathcal{Q}(\mathcal{A},\mathcal{B}) = \sum_{i,j=1}^{d_1,d_2}|\mathrm{Tr}(A_iB_j)|^2 = 1.$$
Since $A_1 = B_1 = \frac{1}{\sqrt{n}}I_n$, we have
$$1 + \sum_{(i,j)\neq (1,1)} |\mathrm{Tr}(A_iB_j)|^2 = 1,$$
and so $\mathrm{Tr}(A_iB_j) = 0$ except when $(i,j)=(1,1)$.

Then, for any $A = \frac{\mathrm{Tr}(A)}{n}I_n + \sum_{i=2}^{d_1}a_i A_i \in \mathcal A$ and $B = \frac{\mathrm{Tr}(B)}{n}I_n + \sum_{j=1}^{d_2}b_j B_j \in \mathcal B$, we have that
\begin{eqnarray*}
\mathrm{Tr}(AB) &=& \frac{\mathrm{Tr}(A)\mathrm{Tr}(B)}{n^2}\mathrm{Tr}(I_n) + \sum_{(i,j)\neq (1,1)}a_ib_j\mathrm{Tr}(A_iB_j) \\
&=& \frac{1}{n}\mathrm{Tr}(A)\mathrm{Tr}(B),
\end{eqnarray*}
which is one of the equivalent conditions for quasiorthogonality.
\end{proof}

\begin{remark} Note that evidently for $\mathcal{A}$, $\mathcal{B}$ unital $*$-algebras, $\mathcal{Q}(\mathcal{A},\mathcal{B}) \geq 1$ since both algebras contain $\frac{1}{\sqrt{n}}I_n.$ So quasiorthogonality corresponds to the case where $\mathcal{Q}$ is minimized.
\end{remark}
Now we end the section with the definition of  $\epsilon$-quasiorthogonal subalgebras.
\begin{definition}
Given an $\epsilon>0$, two unital $*$-subalgebras $\mathcal{A},\mathcal{B}\subset M_n$ are called $\epsilon$- quasiorthogonal if
\[\mathcal{Q}(\mathcal{A},\mathcal{B})\leq 1+\epsilon.\]
\end{definition}

\section{Approximate Relative Quantum Privacy and Main Result}

We shall consider the following notion of approximate privacy in what follows.

We first recall the 2-norm of an operator is $|| A ||_2 = ( \tr (A^* A))^{\frac12}$, and so $< A , A > = || A ||_2^2$. Also $|| \Phi ||_2 = \sup_{||X||_2 = 1} ||\Phi(X)||_2$ for linear maps $\Phi$.

\begin{defn}
Let $\mathcal{A}$, $\mathcal{B}\subseteq M_n(\C)$ be unital $*$-algebras, and let $\epsilon > 0$. Then we say $\mathcal B$ is {\it $\epsilon$-private relative to $\mathcal A$} if
\begin{equation}\label{epsprivate}
|| ( \mathcal E_{\mathcal A}  - \mathcal D_n ) \circ \mathcal E_{\mathcal B} ||_2 < \epsilon ,
\end{equation}
where $\mathcal D_n$ is the complete depolarizing channel on $M_n(\C)$, $\mathcal D_n(X) = \frac{\tr(X)}{n} I_n$ for all $X\in M_n(\C)$.
\end{defn}

\begin{remark}
We are motivated to consider the 2-norm here as it is fairly standard in physically motivated quantum information settings, in addition to the description of $\mathcal Q$ as a particular 2-norm derived in Corollary~\ref{Qcompute}. We also note that our $\epsilon$-private language is in the spirit of terminology used in the context of approximate privacy previously (e.g. \cite{kks}).
\end{remark}

We now state and prove our main result.

\begin{theorem}\label{mainthm} Let $\mathcal{A}$, $\mathcal{B}\subseteq M_n(\C)$ be unital $*$-algebras, and let $d^{\mathrm{max}}_{{\mathcal A} {\mathcal B}} = \max ( \dim {\mathcal A} , \dim {\mathcal B})$ and $d^{\mathrm{min}}_{{\mathcal A} {\mathcal B}} = \min ( \dim {\mathcal A} , \dim {\mathcal B})$. Let $\epsilon > 0$. Suppose that
$$\mathcal{Q}(\mathcal{A},\mathcal{B}) \leq 1 + \frac{\epsilon^2}{d^{\mathrm{max}}_{\mathcal A \mathcal B} - 1}.$$
Then $\mathcal B$ (and respectively $\mathcal A$) is $\epsilon$-private relative to $\mathcal A$ (respectively to $\mathcal B$).

Conversely, suppose that $\mathcal B$ (and respectively $\mathcal A$) is $\epsilon$-private relative to $\mathcal A$ (respectively to $\mathcal B$). Then we have
$$\mathcal{Q}(\mathcal{A},\mathcal{B}) \leq 1 + (d^{\mathrm{min}}_{\mathcal A \mathcal B} - 1)\epsilon^2.$$
\end{theorem}

\begin{proof} As above, we pick o.n. bases $\{A_i\}_{i=1}^{d_1}$ and $\{B_j\}_{j=1}^{d_2}$ with $A_1 = B_1 = \frac{1}{\sqrt{n}}I_n$. (Note that $d^{\mathrm{min}}_{\mathcal A \mathcal B} \leq d_j\leq d^{\mathrm{max}}_{\mathcal A \mathcal B}$, $j=1,2$.) We will also choose the basis so that $A_i = A_i^*$ for each $i$, and similarly for $B_j$.

For the forward direction, let $\epsilon' = \epsilon^2(d_{\mathcal A \mathcal B} - 1)^{-1}$. Since
$$
\mathcal{Q}(\mathcal{A},\mathcal{B}) = 1 + \sum_{(i,j)\neq (1,1)} |\mathrm{Tr}(A_iB_j)|^2 \leq  1 + \epsilon' ,
$$
it follows that for each $j\neq 1$, we have
$$
\sum_{i=2}^{d_1} |\mathrm{Tr}(A_iB_j)|^2 = \sum_{i=2}^{d_1} |< A_i , B_j >   |^2 \leq \epsilon' .
$$

Next, we extend $\{A_i\}_{i=1}^{d_1}$ to an o.n. basis for the full matrix space by appending elements  $\{A_j^\perp \}_{j=1}^{n^2 -d_2}$ and express each $B_j$ in terms of this basis:
$$
B_j = \frac{\mathrm{Tr}(B_j)}{n}I_n + \sum_{i=2}^{d_1} b_{ji} A_i + \sum_{k=1}^{n^2 - d_1} c_{jk}A_k^{\perp}.
$$
By orthonormality of our basis we have $|b_{ji}|^2 = |\mathrm{Tr}(A_i^*B_j)|^2 = |\mathrm{Tr}(A_iB_j)|^2$ for all $i,j$, and hence for all $j\neq 1$,
$$
\sum_{i=2}^{d_1} | b_{ji} |^2 \leq \epsilon' .
$$

Now we apply $\mathcal{E}_{\mathcal{A}}$ to $B_j$, using its decomposition above and noting that $<B_j , I_n> = \tr(B_j) =0$ and $\mathcal E_{\mathcal A}(A_k^\perp)=0$, to obtain
$
\mathcal{E}_{\mathcal{A}}(B_j) =  \sum_{i=2}^{d_1} b_{ji}A_i .
$
Hence by orthonormality of the $A_i$ we have:
\[
\|\mathcal{E}_{\mathcal{A}}(B_j) - \frac{\mathrm{Tr}(B_j)}{n}I_n\|_2^2 = \|\sum_{i=2}^{d_1}b_{ji}A_i\|_2^2
= \sum_{i=2}^{d_1} |b_{ji}|^2  \leq \epsilon' .
\]

Finally, we pick an arbitrary $B\in \mathcal B$ and decompose it as $B = \frac{\mathrm{Tr}(B)}{n}I_n + \sum_{j=2}^{d_2} c_j B_j$. Observe that $\| B \|_2^2 \geq \sum_j |c_j|^2$. Then applying $\mathcal{E}_{\mathcal{A}}$ we get
$$
\mathcal{E}_{\mathcal{A}}(B) - \frac{\mathrm{Tr}(B)}{n}I_n = \sum_{j=2}^{d_2} c_j \mathcal{E}_{\mathcal{A}}(B_j).
$$
Thus we have
\begin{align*}
\big( \|\mathcal{E}_{\mathcal{A}}(B) - \frac{\mathrm{Tr}(B)}{n}I_n\|_2 \big)^2 & = \big( \|\sum_{j=2}^{d_2}c_j\mathcal{E}_{\mathcal{A}}(B_j)\|_2 \big)^2 \\
& \leq \big( \sum_{j=2}^{d_2}|c_j| \|\mathcal{E}_{\mathcal{A}}(B_j)\|_2 \big)^2 \\
& \leq \big( \sum_{j=2}^{d_2} |c_j|^2  \big)  \big( \sum_{j=2}^{d_2} || \mathcal E_{\mathcal A}(B_j)  ||_2^2    \big)  \\
& \leq \|B\|_2^2 (d_2 -1) \epsilon' \\
& \leq \epsilon^2 \| B \|_2^2 .
\end{align*}
As $B\in \mathcal B$ was arbitrary, it follows that $|| ( \mathcal E_{\mathcal A}  - \mathcal D_n ) \circ \mathcal E_{\mathcal B} ||_2 < \epsilon$, and this direction of the proof is complete.

For the converse direction, suppose that $\mathcal B$ is $\epsilon$-private relative to $\mathcal A$, so that
\[
\|\mathcal{E}_{\mathcal{A}}(B) - \frac{\mathrm{Tr}(B)}{n}I_n\|_2 \leq \epsilon \| B  \|_2 \quad\quad  \forall B \in \mathcal B.
\]
As each $B_j$ is traceless and has 2-norm equal to 1, we have $\|\mathcal{E}_{\mathcal{A}}(B_j) \|_2 \leq \epsilon$. So if we write (again using $<B_j, I_n> =0$),
\[
B_j = \sum_{i=2}^{d_1} <B_j , A_i > A_i + \sum_k <B_j , A_i^\perp > A_i^\perp,
\]
then we have
$
\mathcal E_{\mathcal A}(B_j) = \sum_{i=2}^{d_1} <B_j , A_i > A_i
$
and so
\[
\| \mathcal E_{\mathcal A}(B_j)\|_2^2 = \sum_{i=2}^{d_1} | <B_j , A_i >|^2 \leq \epsilon^2 .
\]
Finally, it follows that
\begin{eqnarray*}
\mathcal{Q}(\mathcal{A},\mathcal{B}) &=& 1 +  \sum_{i\neq 1, j\neq 1} | <B_j , A_i >|^2 \\
&=& 1 + \sum_{j=2}^{d_2} \big(   \sum_{i=2}^{d_1} | <B_j , A_i >|^2 \big)   \\
&\leq & 1 + (d_2 - 1)\epsilon^2.
\end{eqnarray*}
Similarly, $\mathcal{Q}(\mathcal{A},\mathcal{B}) \leq 1 + (d_1 - 1)\epsilon^2$, and the result follows.
\end{proof}

\section{Examples}

In this section we apply Theorem~\ref{mainthm} to examples drawn from a number of different quantum information settings.

\begin{example}
	We first present an example fashioned for illustrative purposes, one that also arises in the context of hybrid quantum information memories, processing, and error correction \cite{Kup,bkk1,bkk2,bkk3,braunstein2011zero,grassl2017codes,klappen2018}.

Consider the algebra $\mathcal{A} = M_2(\mathbb{C}) \oplus M_2(\mathbb{C}) \subseteq M_4(\mathbb{C})$ of matrices of the form
	\[ A =
	\begin{pmatrix}
	A_1 & 0 \\
	0 & A_2 \\
	\end{pmatrix},
	\]
with $A_1,A_2 \in M_2(\mathbb{C})$. From a hybrid classical-quantum information perspective, $\mathcal A$ can encode two separate qubits, each with its own classical address.

Additionally, take $\mathcal{B}$ to be the unital $*$-algebra of matrices inside $M_4(\mathbb{C})$ of the form
\[
B = k_1 I_4  +  k_2 \begin{pmatrix}
0&U\\
U^*&0\\
\end{pmatrix}  =
 \begin{pmatrix}
k_1 I_2 & k_2 U\\
k_2 U^* & k_1 I_2 \\
\end{pmatrix}
\]
for complex numbers $k_1, k_2$ and some fixed unitary $U\in M_2(\mathbb{C})$. One can verify that $\mathcal{A}$ is quasiorthogonal to $\mathcal{B}$,  equivalently, $Q(\mathcal{A},\mathcal{B}) = 1 $.

In the specific case that $U = \sigma_x$, the Pauli bit-flip matrix, we now obtain an algebra $\mathcal{C}$ from $\mathcal{B}$ by assuming $\mathcal B$ is exposed to some unitary noise $V$, implemented by the conjugation $C = V B V^*$, $B \in \mathcal{B}$ and $C \in \mathcal{C}$, so $\mathcal C = V {\mathcal B} V^*$. The unitary is chosen to reflect minimal noise exposure, in that the unitary is a small perturbation of the identity, $V = e^T$, where
\[
T = \begin{pmatrix}
\delta&0&0&0\\
0&-\delta&0&0\\
0&0&0&0\\
0&0&0&0\\
\end{pmatrix}
\]
for some fixed $0< \delta  << 1$.

Then computing $Q(\mathcal{A}, \mathcal{C})$ using our characterization from Corollary~\ref{Qcompute} and Mathematica software, we find that
\[
Q(\mathcal{A}, \mathcal{C}) = \frac{1}{2}(\cosh{4\delta} + \cosh{2\delta}).
\]
Thus we may apply the theorem above to quantify the approximate privacy of $\mathcal A$ and $\mathcal{C}$ in the following way:
for any $\epsilon > 0$ satisfying,
\[
\epsilon^2 > 7 \left(\frac{1}{2}(\cosh{4\delta} + \cosh{2\delta})  - 1\right),
\]
we have $Q(\mathcal{A},\mathcal{C}) \leq 1 + \frac{\epsilon^2}{d^{max}_{\mathcal{A}\mathcal{B}} - 1}$, yielding $\epsilon$-privacy as in the theorem. Note that we have used $d^{\mathrm{max}}_{\mathcal{A}\mathcal{C}} = 8$, since $\mathcal{A}$ and $\mathcal{C}$ have dimensions $8$ and $2$ respectively.
\end{example}

\begin{example}
In \cite{jochym1} the first example of a private quantum subsystem \cite{ambainis,boykin,bartlett1,bartlett2} was discovered such that no private subspaces existed for the given channel, and error-correction complementarity \cite{kks,jochym,cklt} failed. This example motivated further work and generalizations, including a framing of it in terms of operator algebra language \cite{ljklp,klp17}. With the algebra perspective we can apply the theorem above to that example.

Here we take $\mathcal{A} = \triangle_4$ to be the algebra of diagonal matrices inside $M_4(\mathbb{C})$ with respect to a given basis. Then let $\mathcal{B} = U^* (I_2 \otimes M_2(\mathbb{C})) U$, where $U$ is the unitary
\[
U = \frac{1}{\sqrt{2}}\begin{pmatrix}
1&0&-i&0\\
0&1&0&i\\
0&1&0&-i\\
1&0&i&0\\
\end{pmatrix} .
\]

Then one can check that we have $\mathcal{Q}(\mathcal{A},\mathcal{B}) = 1$; indeed, this can be seen directly through an application of Corollary~\ref{Qcompute} or as a consequence of the results from \cite{jochym}. Now consider the subalgebra $\mathcal{C} = V^* (I_2 \otimes M_2(\mathbb{C})) V$ defined with the modified unitary $V = Ue^{T}$, where
\[
T = \begin{pmatrix}
\delta&0&0&0\\
0&-\delta&0&0\\
0&0&0&0\\
0&0&0&0\\
\end{pmatrix},
\]
for some $0< \delta << 1$. Hence, $\mathcal C = (e^T)^* \mathcal B (e^T)$.

With a choice of basis for the subalgebras and making use of Corollary~\ref{Qcompute} again, we can compute
\[
\mathcal{Q}(\mathcal{A},\mathcal{C}) = \frac{1}{4} e^{-4 \delta } \left(e^{4 \delta }+1\right)^2.
\]
Then, for some suitable choice of $\epsilon(\delta)$ we have
\[\mathcal{Q}(\mathcal{A},\mathcal{C}) \leq 1 + \frac{\epsilon(\delta)^2}{3},
\]
and by Theorem 3.3 we say that $\mathcal{A}$ is $\epsilon$-private relative to $\mathcal{C}$ (and vice-versa). Note that in this case, $d^{\mathrm{max}}_{\mathcal{A}\mathcal{B}} = 4$.
\end{example}

\begin{example}
Another example of quasiorthogonal subalgebras for which we can study the approximate case comes from the study of mutually unbiased bases (MUB) \cite{bandyopadhyay2002new,kr04,kimura2006,durt2010mutually,spengler2012entanglement}.

MUB are useful in many quantum information protocols because of their defining property. Specifically, two orthonormal basis $\{ |\phi_i\rangle \}$ and $\{|\psi_k\rangle\}$  of $\mathbb{C}^n$ are {\it mutually unbiased} if for all $i,k = 1, \dots , n$,
\[
|\langle \phi_i | \psi_k \rangle| = \frac{1}{\sqrt{n}}.
\]
There is a maximal abelian subalgebra (MASA) \cite{davidson} denoted $\mathcal{A}_{\psi}$ (and similarly $\mathcal{A}_\phi$) in $M_n(\mathbb{C})$ associated with each basis in the following way: $\mathcal{A}_{\psi} = \mathrm{span}\{ P_{\psi_i} : 1 \leq i \leq n  \}$ is the linear span of the orthonomal projectors $P_{\psi_i} = \kb{\psi_i}{\psi_i} \in M_n(\mathbb{C})$ onto the one-dimensional vector subspaces $\mathrm{span}\{\ket{\psi_i} \} = \mathbb{C}|\psi_i \rangle$.
The subalgebras $\mathcal A_\psi$, $\mathcal A_\phi$ are quasiorthogonal if and only if the bases are mutually unbiased. This can be checked using the fact that
\[
\mathrm{Tr}(P_{\psi_i} P_{\phi_k}) =|\langle \psi_i | \phi_k \rangle| ^2
\]
and the criterion for quasiorthogonality, $\mathcal{Q} (\mathcal{A}_{\psi}, \mathcal{A}_{\phi}) = 1$ using Corollary~\ref{Qcompute} for instance.

The maximum number of MUB in an arbitrary dimension is not known in general (they are known for cases where the dimension is a power of a prime). Here we exploit the concept of "approximate" mutually unbiased bases as discussed in \cite{klappenecker2005approximately,shparlinski2006constructions}.

Given an $\epsilon>0$, we call a system of  $n^2+n$ vectors in $\mathbb{C}^n$ 	which are the elements of $n+1$ orthonormal bases $\mathcal{B}_k=\{\psi_{k,1},\cdots,\psi_{k,n}\}$ of $\mathbb{C}^n$ where $k=0,1,2,\cdots,n$ $\epsilon$-approximately mutually unbiased bases if
\[|\langle \psi_{k,i},\psi_{j,l}\rangle|^2\leq\frac{1+\epsilon}{n},\]
for every $0\leq k,l\leq n, k\neq l, 1\leq i,j\leq n.$

One such construction, from \cite{klappenecker2005approximately}, asserts the existence of a system of approximately MUBs with an inequality of the following type:
\[
|\langle \phi_i | \psi_k \rangle| \leq \left(2 + O(n^{-\frac{1}{10}}) \right) n^{-\frac{1}{2}}.
\]
Using the above expression and the measure of orthogonality given in Corollary~\ref{Qcompute}, one sees that the associated MASAs are approximately private and quasiorthogonal for some $\epsilon$. In the 4-dimensional (two-qubit) case, for instance, one has
\[
|\langle \phi_i | \psi_k \rangle| \leq 1 + \lambda(4),
\]
where $\lambda(4)$ is some $O(4^{-\frac{3}{5}})$ expression. Thus each pair of MASAs is $\epsilon$-private whenever $\sqrt{3 \lambda(4)} < \epsilon$ .

\end{example}

\section{Outlook}

We have explicitly linked approximate quasiorthogonality of operator algebras with an appropriate notion of approximate relative privacy for the algebras, determined by the actions of their conditional expectation channels. We focussed on unital algebras and the notion of quantum privacy defined by privatizing to the identity operator as this includes many natural examples and it kept the technical issues manageable. That said, we expect it should be possible to extend this result to more general algebras and more general notions of privacy, for instance as has been accomplished for quantum error correction \cite{bkk3} and private quantum codes \cite{cklt}. It would also be interesting to see if this work could help to generate new constructions of approximate MUB or be applied to the study of SIC-POVM's \cite{klappenecker2005approximately,shparlinski2006constructions} through focus on the commutative algebra case of our result.
We leave these and other investigations to be pursued elsewhere.

\vspace{0.1in}

{\noindent}{\it Acknowledgements.}  D.W.K. was partly supported by NSERC and a University Research Chair at Guelph. J.L. holds a University of Guelph - Institute for Quantum Computing postdoctoral fellowship. M.N. was partly supported by Mitacs and
the African Institute for Mathematical Sciences. R.P. was partly supported
by NSERC. M.R. holds a postdoctoral fellowship in the Department of Pure Mathematics, University of Waterloo.

\bibliography{privacy-mult}
\bibliographystyle{amsplain}

\end{document}